\documentclass[journal]{IEEEtran}

\usepackage{comment}
\usepackage[T1]{fontenc}
\usepackage[utf8]{inputenc}
\usepackage{tensor}
\usepackage[cmex10]{amsmath}
\usepackage{physics}
\usepackage{amssymb,amsfonts}
\usepackage{mathtools, bbm}

\usepackage{amsthm}
\newtheorem{theorem}{Theorem}

\newtheorem*{Game*}{Game}
\newtheorem*{TokGen*}{TokenGeneration Phase}
\usepackage{breqn}
\newtheorem*{TokVer*}{TokenVerification Phase}
\newtheorem*{Setup*}{Setup}
\newtheorem*{Query*}{Query}
\newtheorem*{Challenge*}{Challenge}
\newtheorem*{Guess*}{Guess}
\theoremstyle{remark}

\usepackage{tikz}
\usepackage{tabularx}
\usetikzlibrary{quantikz2}
\usepackage{algorithm}
\usepackage{algpseudocode}
\usepackage{enumitem}
\usepackage{todonotes} %Gerrit

\newcounter{protocol}

\usepackage{graphicx}
\usepackage{url}
\usepackage{cite}
\usepackage{hyperref}
\usepackage[capitalize]{cleveref}

\usepackage{booktabs}
\usepackage{textcomp}

\begin{document}

\title{Privacy-Preserving IoT in Connected Aircraft Cabin}

\author{Nilesh~Vyas,
        Benjamin~Zhao,
        Aygün~Baltaci,
        Gustavo~de~Carvalho~Bertoli, 
        Hassan~Asghar,
        Markus~Klügel,
        Gerrit~Schramm,
        Martin~Kubisch,
        and~Dali~Kaafar% <-this % stops a space
%\thanks{Manuscript received Month Day, Year; revised Month Day, Year. (Corresponding author: Nilesh Vyas.)}% <-this % stops a space
\thanks{Nilesh Vyas, Aygün Baltaci, Gustavo de Carvalho Bertoli,  Markus~Klügel,
        Gerrit~Schramm,  and  Martin~Kubisch, are with Airbus Central R\&T, Germany (e-mail: nilesh.vyas@airbus.com; ayguen.baltaci@airbus.com; gustavo.bertoli@airbus.com; markus.kluegel@airbus.com; gerrit.schramm@airbus.com; martin.kubisch@airbus.com).}% <-this % stops a space
\thanks{Benjamin Zhao, Hassan Asghar, and Dali Kaafar are with the School of Computing, Macquarie University, Sydney, NSW 2109, Australia (e-mail: ben\_zi.zhao@mq.edu.au; hassan.asghar@mq.edu.au; dali.kaafar@mq.edu.au).}}

\maketitle

\begin{abstract}
The proliferation of IoT devices in shared, multi-vendor environments like the modern aircraft cabin creates a fundamental conflict between the promise of data collaboration and the risks to passenger privacy, vendor intellectual property (IP), and regulatory compliance. While emerging standards like the Cabin Secure Media-Independent Messaging (CSMIM) protocol provide a secure communication backbone, they do not resolve data governance challenges at the application layer, leaving a privacy gap that impedes trust. This paper proposes and evaluates a framework that closes this gap by integrating a configurable layer of Privacy-Enhancing Technologies (PETs) atop a CSMIM-like architecture. We conduct a rigorous, empirical analysis of two pragmatic PETs: Differential Privacy (DP) for statistical sharing, and an additive secret sharing scheme (ASS) for data obfuscation. Using a high-fidelity testbed with resource-constrained hardware, we quantify the trade-offs between data privacy, utility, and computing performance. Our results demonstrate that the computational overhead of PETs is often negligible compared to inherent network and protocol latencies. We prove that architectural choices, such as on-device versus virtualized processing, have a far greater impact on end-to-end latency and computational performance than the PETs themselves. The findings provide a practical roadmap for system architects to select and configure appropriate PETs, enabling the design of trustworthy collaborative IoT ecosystems in avionics and other critical domains.
\end{abstract}


\section{Introduction}

The operational paradigms of traditionally isolated, high-stakes environments such as commercial aviation are undergoing a profound transformation, driven by the integration of Internet of Things (IoT) technologies \cite{wang2022intelligent}. The aircraft cabin serves as a canonical example of this new paradigm—a self-contained micro-ecosystem where systems from multiple, often competing, vendors must coexist and operate on shared infrastructure \cite{giertzsch2024design}. This drive towards interconnectedness creates a fundamental duality of collaboration and conflict. On one hand, data sharing can enable collaborative predictive maintenance and highly personalized passenger experiences. On the other hand, this collaboration comes with inherent risks to passenger privacy under regulations \cite{gdpr, ccpa, dpdpa_india, pipeda, pipl_china, lgpd},  and the potential leakage of vendor-specific intellectual property (IP) \cite{rencher2024considerations}.

To address the foundational need for interoperability, the aviation industry has recently introduced a standard called the Cabin Secure Media-Independent Messaging (CSMIM) protocol, specified in ARINC 853 \cite{arinc853}. CSMIM provides a standardized and secure messaging layer, using protocols such as MQTT and Concise Binary Object Representation (CBOR) to simplify the integration of various cabin devices \cite{park2021wireless}. However, CSMIM's security guarantees focus primarily on the transport of data, providing mechanisms for device authentication and access control \cite{arinc853}. It does not address
the privacy of the data content itself. An authorized subscriber can still receive sensitive data in the clear, creating a privacy gap at the application layer that represents a significant barrier to trust and adoption \cite{rencher2024considerations}. This gap is not a flaw in the standard but rather a deliberate design choice that prioritizes establishing a common, secure transport layer as a first step. This reframes the challenge not as fixing a broken standard, but as building the necessary next layer to enable trustworthy data collaboration.

This paper argues that an application-level Privacy-Enhancing Technologies (PETs) layer is essential to close this privacy gap. Furthermore, this layer must be flexible, allowing system architects to navigate the complex requirements of privacy, data utility, and real-time performance. To this end, this work makes the following contributions:
\begin{enumerate}[leftmargin=*]
\item We formalize the multifaceted challenges (technical, business, and regulatory) of multi-vendor data sharing in a constrained environment like an aircraft cabin.
\item We propose a flexible framework for integrating and evaluating distinct classes of PETs, specifically Differential Privacy (DP) and an additive secret sharing scheme (ASS) based on Shamir's secret sharing scheme, within a CSMIM-based architecture.
\item Our comprehensive empirical evaluation, based on a high-fidelity physical testbed that mimics an aircraft cabin, provides novel quantitative data on the utility-privacy trade-off and performance overheads of these technologies in a realistic setting.
\item We deliver a comparative analysis and a set of architectural recommendations to guide the practical design of trustworthy, privacy-preserving IoT systems in aviation and beyond.
\end{enumerate}

\section{Background and Related Work}
\label{sec:background}

\subsection{State-of-the-Art in Cabin IoT and CSMIM}
The evolution of aircraft cabins into smart, data-driven environments requires a shared communication platform to enable novel services like predictive maintenance and operational optimization \cite{rencher2024considerations}. To overcome historical integration challenges from proprietary protocols, the aviation industry developed the Cabin Secure Media-Independent Messaging (CSMIM) standard, specified in ARINC 853 \cite{arinc853}. CSMIM establishes a standardized, broker-based architecture using MQTTv5 to decouple data producers and consumers, enabling multi-vendor interoperability over any network medium, such as Ethernet or Wi-Fi \cite{giertzsch2024design}. The standard provides a secure transport layer with device authentication via mutual Transport Layer Security (TLS) and granular access control using topic-based Access Control Lists (ACLs).

Although specific to aviation, CSMIM’s broker-centric architecture is not unique. It represents a pattern of convergent evolution seen in other critical IoT domains facing similar multi-vendor interoperability issues \cite{privacy-interop,rencher2024considerations}. This architectural paradigm is mirrored in standards like the Industrial IoT's OPC UA and the oneM2M Common Service Layer, which fulfill analogous roles in their respective ecosystems \cite{OPCUA, privacy-interop, oneM2M}.

However, CSMIM’s security is limited to transport and access control, creating a critical privacy gap at the application layer where data content remains unprotected for authorized subscribers. This research is motivated by the need to integrate Privacy-Enhancing Technologies (PETs) to secure the data itself within a CSMIM-compliant system \cite{arinc853}. Due to the architectural similarities across high-assurance IoT systems, the findings on integrating PETs are broadly applicable beyond the aircraft cabin. This elevates this paper's findings from a niche aviation solution to a broadly applicable architectural pattern for securing critical IoT ecosystems.

\subsection{Survey of Privacy-Enhancing Technologies for IoT}
The field of privacy preservation offers a wide array of technologies, each with distinct strengths and weaknesses. 

\subsubsection{Established Cryptographic Techniques} Encryption is the bedrock of data security, used to protect data both in transit and at rest. Protocols like TLS, which underpin secure MQTT communication in CSMIM, rely on asymmetric and symmetric encryption. However, standard encryption requires data to be decrypted before it can be processed, creating a point of privacy leakage. Homomorphic Encryption (HE) provides a solution approach to solves this by allowing computations directly on encrypted data \cite{balle2019privacy, acar2018survey}. Although theoretically ideal, its immense computational overhead currently makes it unsuitable for the resource-constrained and latency-sensitive edge devices in an aircraft cabin.

\subsubsection{Privacy-Preserving Machine Learning Paradigms} Federated Learning (FL) has emerged as a popular paradigm for training a shared machine learning model across distributed clients without centralizing their raw data \cite{yang2019federated, beltran2023decentralized}. However, deploying FL effectively requires that the learning task, its participants, and the data distribution (i.e., a horizontal or vertical setting) are defined in advance. These elements are often not finalized during the early design phase of an aircraft cabin. Furthermore, FL is a specialized paradigm for model training and is less suited as a general-purpose PET for the diverse and often ad-hoc data flows evaluated in our testbed. FL is also not a panacea for privacy; sensitive information can still be leaked through model updates exchanged during training, and it faces significant challenges related to the statistical heterogeneity of data between clients and the communication overhead \cite{kairouz2021advances}.

\subsubsection{Anonymization and Obfuscation Techniques} Traditional data protection methods include k-anonymity and pseudonymization, which aim to prevent re-identification by generalizing or masking identifiers \cite{sweeney2002kanonymity}. While useful in some contexts, these techniques are known to be vulnerable to linkage attacks when an adversary has access to auxiliary information, and they can often degrade data utility to the point of uselessness \cite{balle2019privacy, machanavajjhala2006breaking}. More advanced cryptographic methods like Zero-Knowledge Proofs (ZKPs) allow a party to prove knowledge of a fact without revealing the underlying information \cite{buenz2020survey}. These are powerful techniques but are often tailored to specific verification tasks and can be too complex and computationally intensive for general-purpose data analysis on embedded devices \cite{buenz2020survey}.

\subsection{Positioning Our Work: The Research Gap}
While extensive research exists on individual PETs and their theoretical guarantees \cite{balle2019privacy, sweeney2002kanonymity, buenz2020survey}, a significant gap remains. There is a notable lack of comprehensive, comparative empirical analysis evaluating multiple PETs within a single, realistic framework that mirrors an emerging industry standard like ARINC 853. Much prior work relies on simulations, which often overlook real-world complexities.

This paper bridges that theory-practice gap by providing quantitative performance and utility data from a physical hardware testbed. This approach captures practical challenges like network latency and protocol overhead often abstracted away in theoretical work. Instead of advocating for a single solution, we evaluate a portfolio of pragmatic PETs—DP and an ASS scheme for their feasibility in resource-constrained environments. By doing so, this paper provides not an idealized answer but a nuanced, data-driven decision-making framework for engineers and architects building the next generation of secure and trustworthy private IoT systems.

\section{System Model and Problem Statement}
\label{sec:system}

\subsection{System Architecture and Parties}
We model our system with $n$ sensors, $\mathcal{S}_{i}$, each contributing data $x_i$ to a central server, $C$, which computes a function  $f(x_{1},...,x_{n})$ over the collected data. Communication occurs over one or more of $m$ available broadcast channels,  $\mathcal{B}_{j}$. Within a media-independent protocol like CSMIM, these channels can be realized over any underlying network medium, such as wired ethernet or distinct Wi-Fi channels. This model is instantiated in our experiments using a high-fidelity, three-layered architecture.

\subsubsection{Physical Layer} This layer comprises the hardware, including multiple ESP32 microcontrollers that act as data-producing sensor nodes, representing the diverse, resource-constrained IoT equipment in a cabin. A powerful central server hosts the system's core services, including the data broker and virtualized nodes for tasks like aggregation or machine learning.

\subsubsection{Communication Layer} At the heart of the architecture, a central On-board message broker functions as the message broker. Compliant with ARINC 853, it uses MQTTv5 to manage all data flows in a publish-subscribe model, decoupling producers from consumers and enabling multi-vendor interoperability. For efficiency on embedded devices, all payloads are serialized using the CBOR format.

\subsubsection{Security Layer} This layer is provided by the message broker, which enforces baseline security for the ecosystem. It mandates mutual TLS (mTLS) for all connections, requiring every client to present a valid certificate for authentication. Upon successful authentication, the message broker applies topic-based ACLs, which provide granular control over data flows by defining which message topics a given device is permitted to publish to or subscribe from.

\subsection{Threat Model}
We consider two primary non-colluding adversaries, defined by their capabilities and objectives within the CSMIM ecosystem:

\subsubsection{Honest-but-Curious Aggregator/Vendor} This adversary is a legitimate participant, such as a central server or another vendor's equipment. It correctly follows the communication protocol and possesses valid credentials to authenticate with the message broker. Its access refers to its authorized ability, granted by the message broker's ACLs, to subscribe to one or more MQTT topics. The threat is that this adversary will abuse its legitimate access to analyze the received data and infer the sensitive, raw data contribution $x_i$ from a specific sensor $\mathcal{S}_i$, going beyond the intended purpose of the data stream. This internal adversary is the more complex and commercially significant threat, as it challenges the very premise of data collaboration. The success of the ecosystem hinges on mitigating this internal threat, making PETs a business enabler, not just a security feature.

\subsubsection{Malicious Eavesdropper} This adversary $\mathcal{E}$, is an unauthorized actor physically present within the aircraft cabin who can passively capture wireless network traffic (e.g., 802.11 frames). While mTLS provides a strong primary defense, this adversary aims to exploit any weakness in the transport-level encryption to access the message payloads. We assume the eavesdropper cannot monitor all $m$ broadcast channels (e.g., distinct Wi-Fi channels or wired/wireless paths) simultaneously.

The primary goal of both adversaries is to deduce the raw data of a specific sensor.

\subsection{Data Encoding and Domain}
To handle real-valued sensor data in a digital system, we define a consistent encoding scheme. Let $\mathcal{D}\in\mathbb{R}$ be the domain of values for any data item $x_i$. Let $k$ be a desired precision level. We define: $q_{min} = \min_{x\in\mathcal{D}}\lfloor xk\rfloor$ and  $q = |q_{min}| + \max_{x\in\mathcal{D}}\lfloor xk\rfloor$. Any value $x\in\mathcal{D}$ is then encoded as an integer $x' = |q_{min}|+\lfloor xk\rfloor$, where $x' \in \{0,1,...,q\}$. The corresponding decoding operation for an encoded integer $y$ is $(y-|q_{min}|)/k$.

\subsection{Challenges in Collaborative Cabin Ecosystems}

Deploying a collaborative IoT ecosystem in an aircraft cabin introduces a unique triad of challenges spanning technical limitations, business risks, and stringent regulatory requirements.

\subsubsection{Technical Challenges}
The in-cabin environment imposes significant technical constraints. Avionics IoT devices are typically embedded systems with limited CPU, memory, and power \cite{arinc853, nbiot-aircraft}, making computationally intensive techniques like fully homomorphic encryption impractical for real-time applications \cite{balle2019privacy, acar2018survey}. Furthermore, many non-safety-critical cabin applications, such as crew alerts, still require predictable, low-latency responses, meaning the overhead from any PET must be carefully bounded \cite{dwork2006ourdata, dwork2006calibrating, dwork2014algorithmic}. These constraints create a fundamental tension between privacy and utility; stronger guarantees, like adding more noise in differential privacy, inherently degrade data quality, which can impair predictive models or bias reports. Navigating this trade-off is a central design challenge \cite{zhao2020notone}.

\subsubsection{Business and Intellectual Property (IP) Challenges}
Beyond technical hurdles, significant business concerns inhibit data sharing between vendors. A primary deterrent is the risk of exposing proprietary algorithms, operational data, or unique IP to competitors on a shared data bus \cite{arinc853, rencher2024considerations}. Moreover, in a data-driven economy, a vendor's value is often tied to the data its products generate. The inability to control how data is used after being shared (the Copy Problem) devalues these assets and creates a major disincentive for participation.

\subsubsection{Regulatory and Compliance Challenges}
The aviation industry's stringent regulatory framework extends to data handling. Compliance with regulations \cite{gdpr, ccpa, dpdpa_india, pipeda, pipl_china, lgpd} hinges on the principle that data must not be directly linkable to any single individual. Simultaneously, any new software or hardware is subject to rigorous airworthiness certification processes \cite{arinc853, rtca_do326b}. Consequently, the PETs themselves and the underlying hardware must meet high standards of verifiability and robustness \cite{civilavionics-ch6}.

\section{Technical Approach: The Integrated PETs Layer}
\label{sec:technical}
Our framework integrates a portfolio of PETs, each chosen for its unique privacy properties and its feasibility within the constrained aviation environment.

\subsection{Differential Privacy for Statistical Anonymity}
Differential Privacy (DP) offers a rigorous, mathematical standard for privacy by ensuring an algorithm's output is insensitive to the inclusion or exclusion of any single individual's data \cite{dwork2014algorithmic}. The privacy parameter, $\epsilon$, formally bounds this insensitivity.

\subsubsection{Local Differential Privacy (LDP)} In this model, noise is added directly at the sensor node $\mathcal{S}_i$ before data transmission. Each sensor perturbs its own encoded value $x'_i$ using a DP mechanism (e.g., randomized response or the Laplace mechanism) to produce $x''_i$. This perturbed value is then sent to the server $C$. This provides robust privacy against both the server $C$ and the malicious eavesdropper $\mathcal{E}$.

\subsubsection{Global Differential Privacy (GDP)} In this model, raw (or encoded) data $x'_i$ from multiple devices is collected at a trusted aggregation node. The aggregation function (e.g., sum) is computed, and then noise is added to the final result. This preserves higher utility as noise is scaled to the function's sensitivity, but it requires trust in the aggregator.

\subsection{Additive Secret Sharing for Data Obfuscation in Transit}
Shamir's Secret Sharing is a cryptographic protocol implementing a $(t, n)$ threshold scheme for a secret $s$ \cite{shamir1979how}. The secret is split into $n$ unique shares such that any $t$ shares can reconstruct $s$, but any $t-1$ or fewer shares reveal no information. While Shamir's Secret Sharing provides a robust $(t, n)$ threshold scheme, for our use case, we implement a simpler $(m, m)$ additive secret sharing (ASS) scheme. This approach was chosen for its minimal computational overhead on resource-constrained devices, as it avoids complex polynomial operations. It provides confidentiality against an adversary who cannot access all $m$ shares. For a task like computing the average, each sensor $\mathcal{S}_i$ uses an ASS scheme:
\begin{enumerate}
    \item Let the secret be the encoded value $x'_i$. Let $Q$ be a prime modulus such that $Q > nq$.
    \item For $j=1$ to $m-1$, generate random shares $[[x'_i]]_j \in \mathbb{Z}_Q$.
    \item The final share is computed as $[[x'_i]]_m = x'_i - \sum_{j=1}^{m-1} [[x'_i]]_j \pmod{Q}$.
    \item Each share $[[x'_i]]_j$ is published to a distinct broadcast channel (MQTT topic) $\mathcal{B}_j$.
\end{enumerate}
The server C subscribes to all channels, collects all shares, and sums them to recover the aggregate $\sum x'_i$ for its computation. This protects data in transit and at rest on the untrusted message broker broker.

\section{Theoretical Analysis of Privacy Guarantees}
\label{sec:theory}
To ground our experimental results, we first formalize the privacy guarantees of the selected PETs.

\subsection{Quantifying Privacy Leakage in DP via Hypothesis Testing}
Interpreting the $\epsilon$ parameter of DP can be non-intuitive. We use a hypothesis testing framework to provide a concrete meaning. Consider a strong adversary who knows a dataset $X=(x_1, ..., x_{n-1}, x_n)$ except for the value of $x_n$, which they know is one of two values, $s$ or $t$ (with $s<t$). The adversary forms two competing hypotheses: $H_s: x_n = s$ and $H_t: x_n = t$.

Let the query be a simple sum, $q(X) = \sum x_i$, with sensitivity $\Delta q = \max_{x \in \mathcal{D}} |x| = \Delta x$. The DP mechanism releases a noisy answer $z = q(X) + Lap(\Delta q / \epsilon)$. The adversary observes $z$ and uses a likelihood ratio test to decide between $H_s$ and $H_t$:
\[
\frac{Pr[z|H_{t}]}{Pr[z|H_{s}]} = \frac{f_{Lap}(z-q(X_{t}))}{f_{Lap}(z-q(X_{s}))} = \frac{\exp\left(-\frac{\epsilon|z-q(X_{t})|}{\Delta q}\right)}{\exp\left(-\frac{\epsilon|z-q(X_{s})|}{\Delta q}\right)}
\]
Following the Neyman-Pearson lemma \cite{cover2006elements}, the optimal test is to decide $H_t$ if $z > \tau$, where the threshold $\tau = \sum_{i=1}^{n-1}x_{i}+\frac{s+t}{2}$. The probability of the adversary guessing correctly, $Pr[G]$, can be shown to be:
\begin{align}
Pr[G] = Pr\left[Lap(\Delta q/\epsilon)\le\frac{t-s}{2}\right] 
= 1-\frac{1}{2}\exp\left(-\frac{\epsilon(t-s)}{2\Delta q}\right) \label{eq:guess_prob}
\end{align}
This equation provides a direct link between the abstract parameter $\epsilon$ and the concrete probability of an adversary's success. A probability close to $0.5$ (random guessing) indicates strong privacy. This analysis applies to both LDP and GDP models, differing only in the context of the adversary.

\subsection{Information-Theoretic Security of the additive Secret Sharing}
The security of ASS is information-theoretic.  For an $(m, m)$ ASS scheme over $\mathbb{Z}_Q$, any coalition holding $m-1$ or fewer shares has zero information about the secret in the Shannon sense. This information-theoretic security holds provided the shares are chosen uniformly at random. Our implementation achieves this by ensuring that the first $m-1$ shares are random, which makes the final share deterministically random as well. To formalize the utility aspect in our context, we analyze an ASS scheme. Let data from $n$ sensors be encoded as integers $x'_i \in \{0, ..., q\}$. We wish to compute the average.
\begin{theorem}
Let each sensor $S_i$ split its encoded value $x'_i$ into $m$ shares over a finite field $\mathbb{Z}_Q$ where $Q > nq$. Let the server C collect all $nm$ shares, sum them modulo $Q$, divide by $n$, and decode. If $\widehat{ave}$ is the computed average and $(\sum x_i)/n$ is the true average, then
\[
\left|\frac{\sum_{i=1}^{n}x_{i}}{n}-\widehat{ave}\right|<\frac{1}{k}
\]
where $k$ is the precision level of the initial encoding.
\end{theorem}
\begin{proof}
Since $Q > nq \ge \sum x'_i$, the summation of shares modulo $Q$ is equivalent to the true sum of encoded values, $\sum x'_i$. The decoding process introduces a quantization error for each $x_i$ bounded by $1/k$. The error on the average is therefore also bounded by $1/k$.
\end{proof}
This demonstrates that ASS can be used for exact computation on encoded data while providing strong confidentiality against an adversary who cannot access a threshold number of shares.

\subsection{Mapping PETs to Threats}
The chosen portfolio of PETs provides a multi-layered defense tailored to our adversaries.

\subsubsection{Countering the Honest-but-Curious Aggregator/Vendor} Differential Privacy is ideal for statistical tasks, allowing the aggregator to receive a useful result while providing plausible deniability for any individual's contribution. ASS addresses scenarios where even the central infrastructure is untrusted. By splitting a secret across multiple MQTT topics, our scheme ensures the message broker broker itself never holds the complete data.

\subsubsection{Countering the Malicious Eavesdropper} While TLS is the primary defense, PETs provide a critical second layer. Local Differential Privacy is particularly effective, as the data is perturbed before transmission. Even if an eavesdropper breaks TLS encryption, the captured payload is already noisy. ASS significantly raises the bar for the eavesdropper, requiring them to compromise and correlate messages from multiple distinct topics to reconstruct the secret.

\section{Experimental Evaluation}
\label{sec:eval}

\subsection{Experimental Setup}

\subsubsection{Hardware Testbed} The physical testbed was constructed using commercially available hardware to mirror a realistic deployment environment. Sensor nodes, representing data producers, were implemented using ESP32-WROVER-DevKit-Lipo and ESP32-POE-ISO microcontrollers. These devices were chosen for their built-in wireless connectivity and sufficient processing power for performing on-device PETs. A Raspberry Pi 4B served as the primary data consumer and measurement device. A central computing server hosted the message broker and all the virtualized processing nodes as depicted in Figure~\ref{fig:timing-measure}.

\subsubsection{Software and Protocols}
The firmware for the ESP32 devices was compiled through the ESPHome framework, extended with custom C++ components for DP and ASS algorithms. The virtualized data processing pipeline was managed by a custom framework that ingests a GML file to deploy Docker containers. All communication used MQTTv5 with CBOR-encoded payloads, as specified by CSMIM \cite{arinc853}.

\subsubsection{Measurement Methodology}
To ensure accurate, cross-device timing, a hardware-based synchronization method was employed. Standard protocols like Simple Network Time Protoco (SNTP) had unacceptably high variance. Instead, the initiating ESP32 sets a GPIO pin at the start of its processing. This pin was wired to the Raspberry Pi (RPi), which records a start timestamp. Upon receipt of the final message, the RPi records an ending timestamp. By referencing only one clock, this method provides an accurate end-to-end latency measurement.

We acknowledge that the ESPHome framework adds a consistent computational overhead across all our experiments (See Table \ref{tab:pet_overhead}). Although custom-optimized firmware would result in lower absolute latencies, this uniform overhead does not affect the validity of our comparative results as it is applied uniformly to all experimental conditions.

\begin{figure}[ht]
    \centering
    \includegraphics[width=1\linewidth]{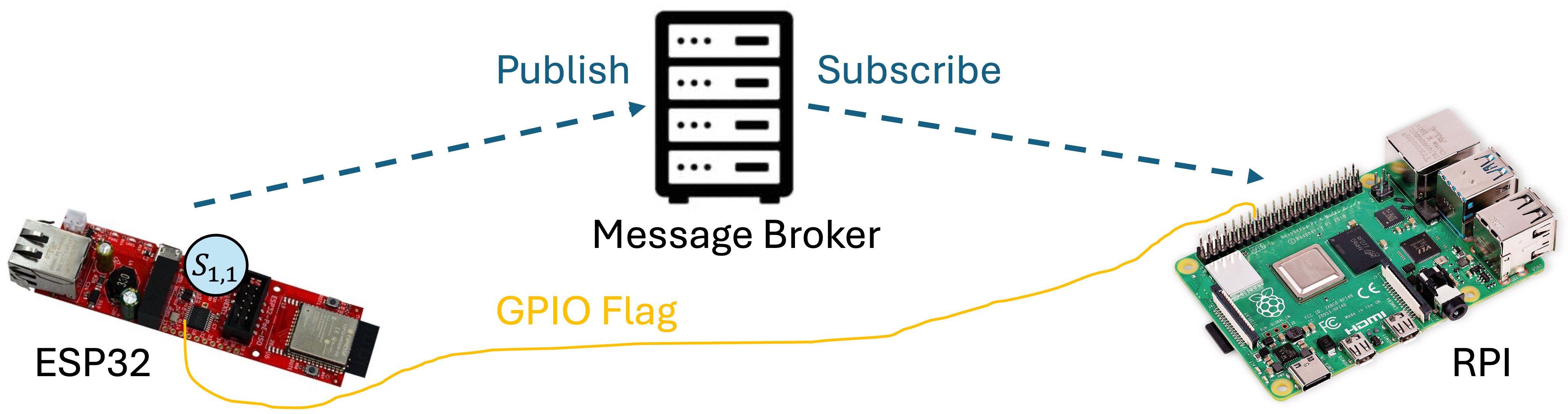}
    \caption{Diagrammatic connections of the ESP32, message broker and RPI for data over MQTTv5 and a physical timing flag used in our timing measurement methodology.}
    \label{fig:timing-measure}
\end{figure}

\subsection{Performance Analysis: Deconstructing System Latencies}
To accurately assess the overhead introduced by various Privacy-Enhancing Technologies (PETs), it is crucial to first deconstruct the inherent latencies of the system. This involves establishing a performance baseline and understanding the costs associated with the network, the central broker, and the architectural complexity of the data pipeline.

\subsubsection{On-Device vs. Virtualized Environments}
Our experiments compare two distinct processing environments, reflecting practical architectural choices in an aircraft cabin:

\textit{On-Device Environment:} Computations are performed directly on the resource-constrained ESP32 microcontroller that is co-located with the sensor. This represents an edge computing paradigm, where data is processed at the source to minimize latency and preserve privacy by preventing raw data from leaving the device. This is typical for systems where sensors and their initial processing units are integrated, such as in a smart seat.

\textit{Virtualized Environment:} Computations are performed within Docker containers hosted on a powerful, centralized server. This represents a more flexible architecture where services can be deployed and scaled dynamically. However, it requires raw data to be transmitted from a sensor to the central server, introducing additional network hops and potential privacy risks if the data is not otherwise protected in transit.

\subsubsection{Baseline Network and Broker Overheads}
The network and the central broker introduces their own delays, which must be understood before evaluating the PETs themselves. As detailed in Table \ref{tab:ping_latencies}, the physical network medium is primary source of latency. A standard Wi-Fi connection between an ESP32 and a target laptop (with no power saving) exhibited an average ping time of 8.0 ms. In contrast, direct wired Ethernet  reduces latency to under 1.0 ms.  More importantly, the central On-board message broker, a necessary component for interoperability and data governance, imposes a non-trivial delay. A single message round-trip through the broker, even over a direct wired Ethernet connection, incurred an average latency of 4.52 ms. This broker overhead is a significant and unavoidable component of any transaction, constituting a substantial portion of the latency budget for any real-time application.

\begin{table}[ht!]
\centering
\caption{Baseline System Latencies of the network when the ESP devices and a target laptop is connected under various power settings and connection types.}
\begin{tabular}{@{}rccc@{}}
\toprule
\textit{Connection Types} & \textit{Min Ping} & \textit{Avg Ping} & \textit{Max Ping} \\
\textit{(Target is the laptop)} & (ms) & (ms, n=100) & (ms) \\ \midrule
High-Power Saving (Wifi-Target) & 31 & 355 & 809 \\
Light-Power Saving (Wifi-Target) & 110 & 270 & 658 \\
No Power Saving (Wifi-Target) & 4 & 8 & 34 \\
Eth-ESP (Eth-Target) & <1 & <1 & 1 \\
Eth-ESP (Wifi-Target) & 2 & 4 & 15 \\ \bottomrule
\end{tabular}
\label{tab:ping_latencies}
\end{table}

\subsubsection{The High Cost of Communication Hops}
To isolate and quantify the latency added by architectural complexity, we measured the delay of a message passing through a chain of five virtual relay nodes hosted on the central server. Each relay simply receives a message from the message broker and immediately forwards it to the next node in the message chain. This experiment, which involves 12 distinct network hops (one publish from the source, five publish/subscribe pairs for the relays, and one subscribe at the destination), revealed an average end-to-end latency of 46.55 ms over 350 repeated runs. Since the processing time on each relay node is negligible, this substantial delay is almost entirely attributable to the repeated interactions with the central broker. This result starkly illustrates that the number of communication hops is a dominant factor in overall system latency, far outweighing the computational cost of many on-device tasks.

\subsection{Performance Overhead of Privacy-Enhancing Technologies}
With an understanding of the baseline latencies, we now analyze the specific overheads of each PET. The results, summarized in Table \ref{tab:pet_overhead}, reveal that the performance impact of a PET is determined not just by its computational complexity, but by where it is executed and how it interacts with the communication protocol.

\subsubsection{Scenario 1: Local Differential Privacy (On-Device)}
In this scenario, the ESP32 generates a value and applies Local Differential Privacy (LDP) before publishing it to the message broker. Our results show that the baseline operation (no processing) has a mean end-to-end latency of 5.295 ms (median 4.8357 ms, std. dev. 1.7347 ms), with the on-device computation taking just 0.307 ms. Applying LDP ($\epsilon$=0.01) increases the on-device compute time to 0.4875 ms. As summarized in Table \ref{tab:pet_overhead}, this represents a raw computational overhead of only 0.18 ms, demonstrating that the core LDP algorithm is lightweight and computationally feasible for resource-constrained embedded systems. While this is a 58.8\% increase in on-device processing time, its effect on the end-to-end latency is modest, increasing it to 6.2057 ms. This shows that for LDP, the network and broker latencies remain the dominant factors.

\subsubsection{Scenario 2: Global Differential Privacy (On-Device vs. Virtualized)}
This scenario compares two architectural patterns for aggregation and the application of Global Differential Privacy (GDP).

\textit{On-Device GDP:} This setup simulates a scenario where a single edge device (e.g., a zonal controller in a seat row) acts as a trusted aggregator for several simpler, nearby sensors. It collects their raw data, computes the aggregate, and applies GDP before publishing the result, thereby isolating the performance impact of the on-device computation and a single MQTT publication. The on-device computation for a simple average is 1.021 ms, which increases to 1.147 ms when GDP ($\epsilon$=1) is added. The full end-to-end latency for this on-device aggregation is 7.095 ms.
 
\textit{Virtualized GDP:} Three separate ESP32 devices publish raw values to the message broker. A virtual node on the central server subscribes to these topics, performs the aggregation and GDP, and publishes the final result. While the compute time on the powerful virtual node is negligible, the end-to-end latency balloons to 13.551 ms.

The \textasciitilde6.5 ms difference between the on-device and virtualized approaches is entirely attributable to the extra communication hops required for the virtualized model (Sensor $\rightarrow$ message broker $\rightarrow$ Virtual Node $\rightarrow$ message broker $\rightarrow$ Consumer). This provides a clear, quantitative demonstration of the significant performance benefit of performing computations at the edge, a critical consideration for systems with real-time constraints.

\subsubsection{Scenario 3: Additive Secret Sharing (On-Device vs. Virtualized)}
This scenario highlights how protocol interactions can dominate performance.

\textit{On-Device ASS:} An ESP32 splits a value into three shares and publishes each to a separate topic. The on-device computation is very fast, averaging just 0.591 ms. However, the end-to-end latency is 14.995 ms. This is because our implementation sequentially publishes the three shares, incurring the broker's round-trip overhead multiple times. This is a straightforward implementation; a more optimized approach using asynchronous MQTT clients could potentially parallelize the publishing of shares, significantly reducing this protocol-induced latency. However, this highlights the sensitivity of such protocols to implementation choices.
  
\textit{Virtualized ASS:} When the secret sharing is moved to a virtual node, the latency increases dramatically. The virtual node itself takes a significant 17.265 ms to perform the operation (which includes receiving the input and sending three outputs), and the total end-to-end latency reaches 35.707 ms.

This scenario illustrates that system performance is often dictated more by the number and nature of protocol interactions with the central broker than by the raw computational complexity of the privacy algorithm itself. The scheme guarantees reconstruction of the original sum if all shares are received, ensuring maximum utility. However, this method is inherently fragile. The loss of even a single share due to network failure results in a total inability to recover the secret. In large-scale systems, this risk of a single point of failure increases, threatening the entire aggregation.

\subsubsection{Scenario 4: System Scalability Under Load}
To test the robustness of the system, we re-ran the LDP scenario while subjecting the message broker to a high load of 400 additional messages per second (to represent an real cabin scenario). The results showed no significant degradation in either the on-device compute time (0.4875 ms vs. 0.4476 ms) or the end-to-end latency (6.2057 ms vs. 5.3893 ms). This indicates that, when adequately provisioned, the central broker is not a bottleneck under the tested conditions, reinforcing the conclusion that latency is primarily driven by the network medium and the number of communication hops in the data path.

\begin{table}[ht!]
\centering
\caption{Performance Overhead of PETs (On-Device vs. Virtualized).}
\begin{tabular}{@{}rcc@{}}
\toprule
\textit{PET Configuration} & \textit{Node Compute}  & \textit{End-to-End} \\ 
 & \textit{Time (ms)} & \textit{Latency (ms)} \\ \midrule
No Processing (Baseline) & 0.307 & 5.295 \\
LDP (On-Device) & 0.488 & 6.206 \\
GDP (On-Device) & 1.147 & 7.095 \\
GDP (Virtualized) & (Negligible) & 13.551 \\
ASS (On-Device) & 0.591 & 14.995 \\
ASS (Virtualized) & 17.265 & 35.707 \\ \bottomrule
\end{tabular}
\label{tab:pet_overhead}
\end{table}

\subsection{Utility Analysis: Quantifying the Privacy-Utility Trade-off}
The following experiments quantify the fundamental \textit{Privacy-Utility Trade-off} for Differential Privacy.

\subsubsection{Statistical Tasks - Aggregate Load Calculation}
This first scenario simulates a practical statistical task: Calculating the total weight of passengers on an aircraft for fuel load estimation. A single ground truth dataset of 500 passenger weights was created by uniformly sampling a value between 50 and 120 Kg, with a resulting total sum of 43,007 Kg. Local Differential Privacy (LDP) in varying levels of the privacy parameter, $\epsilon$ was applied to each individual's weight before a final sum. The utility was measured as the absolute error between the noisy, yet privacy-preserving sum and the true sum.
    
Figure~\ref{fig:sum_weight}, demonstrates the direct and predictable relationship between $\epsilon$ and error against the true sum to provide a concrete tool for system designers. For example, if an airline's operational tolerance for fuel calculation is a 50 Kg error (a mere 0.116\% of the total), the data shows this can be achieved with an $\epsilon$ value of approximately 0.3. This directly links the abstract mathematical parameter of privacy to a tangible business requirement.

\begin{figure}[ht]
    \centering
    \includegraphics[width=0.9\linewidth]{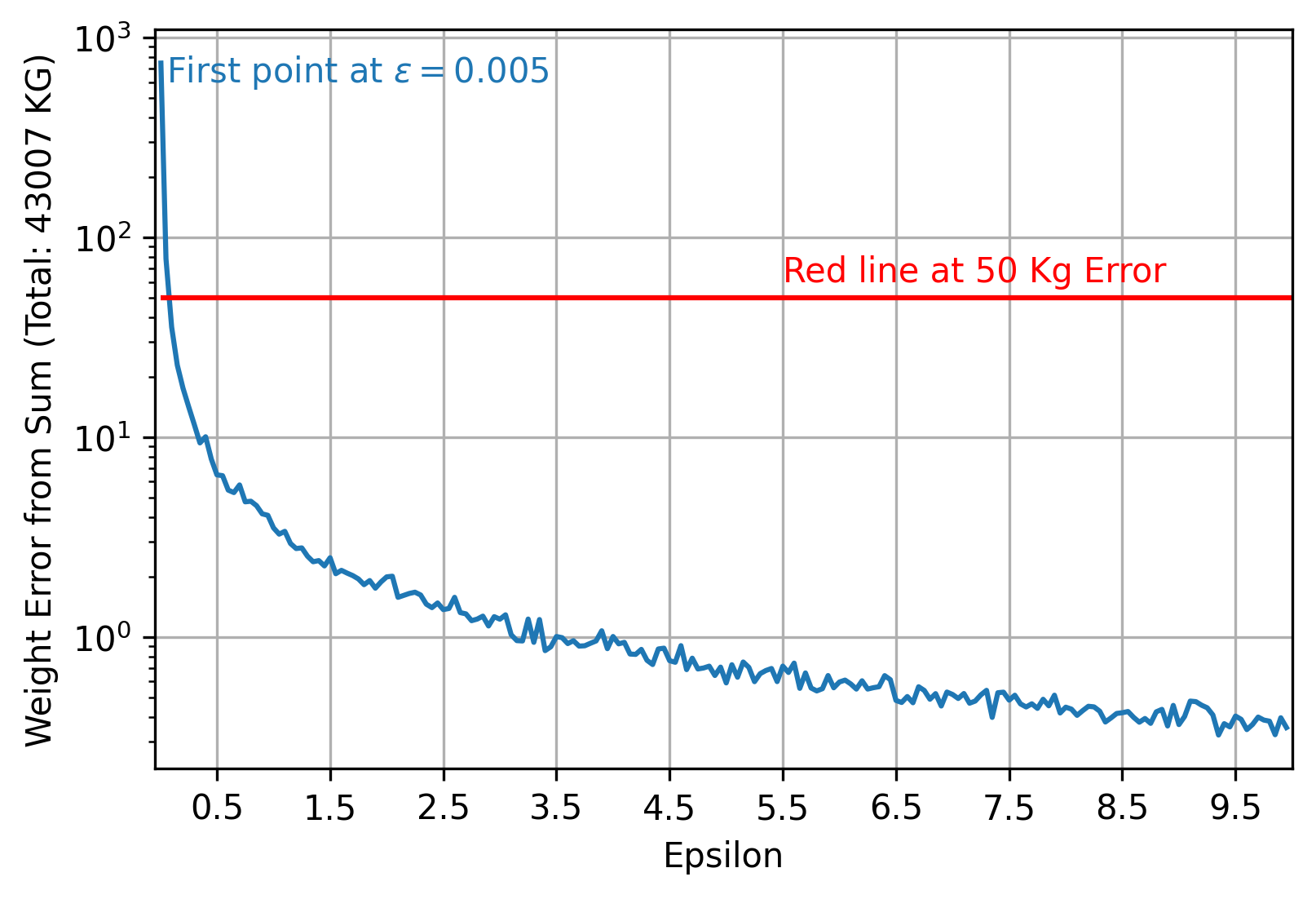}
    \vspace{-4mm}
    \caption{Relationship between $\epsilon$ and error against the true weight sum. In a high privacy setting (low $\epsilon$, e.g., 0.05), the error is substantial, on the order of 100 Kg. As the privacy guarantee is relaxed (higher $\epsilon$), the error decays exponentially, falling to 10 Kg for just before $\epsilon = 0.5$. }
    \label{fig:sum_weight}
\end{figure}

\subsubsection{Privacy-Preserving Inference Task - Seatbelt Detection ML Model}

We evaluate a system for classifying passenger seatbelt status, a critical component for automotive safety applications. The goal is to infer one of three states from sensor data: (1) the seat is unoccupied, (2) a passenger is present but not wearing a seatbelt, or (3) a passenger is present and wearing a seatbelt. The raw X, Y, and Z-axis accelerometer data was sourced from an experimental setup presented by Goyal et al. \cite{Goyal2023PassengerSafety}. 
The raw accelerometer data is a primary privacy risk because this biometric signal can be used to infer highly sensitive personal information, including biometric ID from a unique vibrational fingerprint, health data from micro-movements like tremors or breathing, and physical attributes like weight or BMI. 

A 3-layer  Convolutional Neural Network (CNN) with 25 hidden neurons on each layer trained to classify accelerometer data into three-classes as discussed above. The baseline model trained on a clean 80/20 train/test split achieved a high accuracy of 98.28\%.  Following this baseline, our methodology evaluates the model's utility in a real-time inference scenario known as the Untrusted Scrutinizer model. This simulates a common and valid deployment pattern where a service provider running the model cannot be trusted with a user's raw, sensitive data. To protect the user, the raw sensor data is privatized using Local Differential Privacy (LDP) on the user's device before it is ever sent for classification.

With strong privacy (low $\epsilon$), the model's performance collapses. At $\epsilon=0.5$, accuracy plummets to around 65\%, rendering a model useless for any practical safety or service application. As $\epsilon$ increases, accuracy recovers, crossing a hypothetical acceptable threshold of 95\% at an $\epsilon$ value of approximately 4.0 before plateauing near the original baseline. This highlights the critical need for careful tuning of $\epsilon$ in ML contexts. Too strong of a privacy guarantee destroys utility, while too large an $\epsilon$ may not provide any meaningful privacy guarantee, corresponding to a higher probability of an adversary successfully inferring a user's input data \cite{arinc853}. It should be noted that while higher values of $\epsilon$ (e.g., > 4.0) improve utility, they offer progressively weaker privacy guarantees, and their suitability must be carefully weighed against the specific threat model and regulatory requirements.

\subsubsection{Intellectual Property Protection - Coffee Machine Brewing Profile}
The final scenario considers a high-end coffee machine's unique temperature profile during a brewing cycle as valuable IP to be protected. A competitor with access to the CSMIM bus could analyze this data to reverse-engineer the process. To counter this, the vendor can apply LDP to the temperature sensor data before publishing.

The choice of $\epsilon$ becomes a strategic business decision. A low $\epsilon$ (e.g., 0.01) adds significant noise, completely obscuring the unique signature of the brewing cycle. This provides strong IP protection but may mask subtle thermal anomalies required for predictive maintenance. Conversely, a high $\epsilon$ (e.g., 1.0) adds minimal noise, preserving the profile's utility for detailed diagnostics but risking the exposure of the proprietary curve shape to a competitor. This scenario exemplifies the direct trade-off between IP protection and data utility for maintenance.

\section{Discussion and Conclusion}
\label{sec:conclusion}
This paper has demonstrated that while standards like ARINC 853 provide a secure foundation, a configurable application-level layer of PETs is essential to bridge the privacy gap in multi-vendor IoT ecosystems. Our empirical evaluation on a high-fidelity hardware testbed provides novel, quantitative data on the real-world trade-offs between privacy, utility, and performance. The findings show that it is feasible to integrate these technologies into a modern aviation data framework, with performance being dictated more by architectural choices and network hops than by the computational overhead of the PETs themselves. 

The experimental results reveal that selecting and deploying PETs is a complex engineering trade-off that leads to several key architectural principles. The performance results create a new design paradigm for real-time private systems: the focus must shift from micro-optimizing the PET algorithm's CPU cycles to macro-optimizing the data flow architecture to minimize network hops and protocol interactions. No single PET is universally best; Differential Privacy is ideal for statistical analysis due to its low overhead but directly impacts utility, while the implemented ASS scheme provides strong confidentiality but incurs significant latency from its protocol interactions.

The results underscore two critical architectural principles. First, the primacy of physical topology: hardware layout decisions made early in the design lifecycle have lasting implications for privacy capabilities. Second, for latency-sensitive tasks, on-device processing is significantly faster than virtualized processing due to the high cost of network hops. The investigation also highlights a crucial distinction between privacy  and secrecy, suggesting that the future of trustworthy data collaboration likely lies in hybrid models.

For future investigation, one path is to explore dedicated cryptographic hardware in embedded systems to reduce the overhead of more intensive PETs like Homomorphic Encryption. Another is to design systems capable of adjusting PET settings dynamically in response to changes in network performance, trust, or data importance. It is also vital to collaborate with regulatory and standards organizations to establish a clear certification process for PET-equipped systems in commercial aviation. Lastly, a critical next step is to empirically validate the scalability of these protocols, evaluating how the end-to-end latency of multi-party schemes scales as the number of participating sensors, $n$, increases from tens to hundreds.

By addressing these challenges, the industry can move closer to realizing the full potential of the connected aircraft cabin, building ecosystems that are not only intelligent but also secure, private, and trustworthy by design.

\bibliographystyle{IEEEtran}
\bibliography{ref.bib}
\end{document}